\documentclass[letterpaper, 10 pt, journal, twoside]{IEEEtran}
\IEEEoverridecommandlockouts                            

\usepackage{caption}
\usepackage{blindtext, graphicx}

\usepackage{epsfig} 
\usepackage{times} 
\usepackage{graphics} 
\usepackage{epsfig} 
\usepackage{times}
\usepackage{amsmath} 
\usepackage{color}
\usepackage{fix-cm}
\usepackage{mathtools}
\usepackage{fix-cm}
\usepackage{soul}
\usepackage{subcaption}

\usepackage{caption}
\usepackage{comment}
\usepackage{xcolor}
\usepackage{cite}
\usepackage[colorlinks=true, linkcolor=blue, citecolor=green, urlcolor=blue]{hyperref} 
\usepackage[linesnumbered,ruled,vlined]{algorithm2e}

\usepackage{amssymb}

\usepackage{amsthm}
\usepackage{float}
\theoremstyle{remark}

\newtheorem{remark}{Remark}
\newtheorem{theorem}{Theorem}
\newtheorem{lemma}{Lemma}
\newtheorem{definition}{Definition}

\usepackage{algorithmic}
\usepackage{xcolor}
\usepackage{comment}

\usepackage{amsmath} 

\usepackage{hyperref}
\usepackage{xcolor}

\usepackage[linesnumbered,ruled,vlined]{algorithm2e}

\usepackage{algorithmic}
\newif\ifLongVersion 
\LongVersiontrue
\newcommand{\COLOR}{black} 

\allowdisplaybreaks

\title{\LARGE \bf
\textcolor{\COLOR}{Stability and Convergence Analysis of Multi-Agent Consensus with Communication Delays: A Lambert W Function Approach}}

\author{Layan Badran, Kiarash Aryankia, and Rastko R. Selmic, \IEEEmembership{Senior Member,~IEEE}
\thanks{L. Badran, K. Aryankia, and R. R. Selmic are with the Department of Electrical and Computer Engineering, Concordia University, Montreal, QC, Canada.
       {\tt\small layan.badran@mail.concordia.ca},
       {\tt\small kiarash.aryankia@concordia.ca},
       {\tt\small rastko.selmic@concordia.ca}}
}

\allowdisplaybreaks

\begin{document}

\maketitle
\thispagestyle{empty}
\pagestyle{empty}

\begin{abstract}
This paper investigates the effect of constant time delay in weakly connected multi-agent systems modeled by double integrator dynamics. A novel analytical approach is proposed to establish an upper bound on the permissible time delay that \textcolor{\COLOR}{ensures stability and consensus convergence}. The analysis employs the Lambert W function method in higher-dimensional systems to derive explicit \textcolor{\COLOR}{conditions under which consensus is achieved. The theoretical results are rigorously proven and provide insight into the allowable delay margins. The analysis applies to general leaderless undirected network topologies. The framework also accounts for complex and realistic delays, including non-commensurate communication delays. Numerical examples are provided to demonstrate the effectiveness of the proposed method.}
\end{abstract}

\section{Introduction}\label{sec1}

This paper examines communication delays in a connected leaderless network. When modeling real-world applications, delays are introduced to account for uncertainties arising from environmental disturbances affecting sensor inputs. \textcolor{\COLOR}{Particularly, communication delays of non-commensurate type are considered: each agent is assumed to have access to its own state information instantaneously through embedded sensors, but receives delayed information from its neighbors.}  While non-commensurate communication delays are mathematically complex, they more accurately capture real-world scenarios. Therefore, this analysis focuses exclusively on non-commensurate delays. The proposed methodology is also applicable to commensurate delays and, in fact, significantly simplifies their analysis.

The critical time delay of multi-agent systems is modeled by delay differential equations and was extensively studied by \cite{16}. The authors investigated the correlation between the delay value and system stability, providing a sufficient condition for stability based on the critical delay value. \textcolor{\COLOR}{In \cite{Sumpter}, a delay bound is derived for the second order case}. However, their analysis is limited to commensurate delays. \textcolor{\COLOR}{The authors in \cite{review} and \cite{recent} provided a detailed review of the major advances in the field, explaining how delays impact different systems' stability.} In many cases, while delay affects convergence time, stability itself remains independent of delay. Complete stability is harder to analyze since equations of this type have infinite solutions.

An infrequently explored method for addressing the issue of infinite solutions involves the utilization of the Lambert W function. \textcolor{\COLOR}{The Lambert W function is defined as any function satisfying $x=W(x)e^{W(x)}$ and} has traditionally been used for scalar systems and, to a lesser extent, higher-dimensional systems in specific canonical forms (\hspace{1sp}\cite{Moore,sunyi,HWANG20051979}). Its application to multi-agent systems remains limited, largely due to the mathematical complexity of defining matrix functions. However, symmetry, particularly in Laplacians derived from undirected graphs, simplifies such analysis. The Lambert W function is thus especially well-suited for multi-agent systems with symmetric Laplacians, making the results broadly applicable.

\textcolor{\COLOR}{In the field of multi-agent systems, \cite{Kia} uses the Lambert W function to obtain an estimate of the worst case rate of convergence to consensus. In \cite{Sumpter}, the effect of delay on multi-agent consensus is experimentally analyzed with the Lambert W function using simulations. Both authors in \cite{Kia} and \cite{Sumpter} conclude that a purely analytical approach is highly desirable and the need of generalization to higher order systems. In addition, a self-triggered consensus approach using Lambert W-based sampling for single-integrator multi-agent systems is proposed in
\cite{wakaiki2024self}.}

This paper analyzes the consensus of the multi-agent system using a matrix-based Lambert W function approach. Moreover, the critical delay margin is identified analytically, theoretically proven, and verified with a numerical example.

This paper presents two key contributions:
\begin{enumerate}  
   \item Extending the application of the Lambert W function to multi-agent systems and establishing a relationship between the delay parameter and the system's stability. In this regard, the results are novel.  
   \item Explicitly deriving the system eigenvalues as functions of the delay parameter, along with a necessary and sufficient condition for stability.
\end{enumerate}  
\ifLongVersion{} \else \textcolor{\COLOR}{Due to space constraints, some details were omitted in this manuscript. A comprehensive version is available at \cite{Arxiv}.} \fi

The paper is organized as follows: Section \ref{sec2} covers the necessary background in graph theory and the Lambert W function. Section \ref{sec3} presents the problem formulation. Section \ref{sec4} contains an extensive mathematical proof of the time delay bound. Section \ref{sec5} presents numerical results. Finally, Section \ref{sec6} is a conclusion and discussion of future work.

\section{Background}\label{sec2}

\subsection{Notation}Let ${\mathbb{R}^n}$ denote an $n$-dimensional vector space, and ${\mathbb{R}^{n \times m}}$ denote the collection of $n \times m$ real matrices. Let $\mathbb{C}$ represent the set of complex numbers, and $\mathbb{C}^{n \times n}$ denote the collection of $n \times n$ complex matrices. The $n \times n$ identity matrix is denoted by $\mathbf{I}_n$, and an $n$-vector of all ones is denoted by $\mathbf{1}_n = [1, \dots, 1]^T$. A diagonal matrix constructed from a vector $x$ using the $\text{diag}$ operator is denoted by $X = \text{diag}(x)$. Bold notation is reserved for matrices. For a matrix $\mathbf{A}$, $\lambda_i(\mathbf{A})$ denotes the eigenvalue of $\mathbf{A}$, and $\text{Re}[\lambda]$ denotes the real part of the eigenvalue $\lambda$.
For a matrix $\mathbf{A}$, its nullity space is denoted by $\mathrm{Nullity}(\mathbf{A})$. The symbol $\cup$ denotes the union of two sets. The determinant of a square matrix $\mathbf{A} \in \mathbb{R}^{n \times n}$ is denoted by $\det(\mathbf{A})$.

\subsection{Graph Theory}
Consider a multi-agent system modeled by an undirected graph $G = (V, E)$ with $n$ agents. Assume that the graph is connected: the agent set is $V = \{1, 2, \dots, n\}$, and the edge set $E \subseteq V \times V$ represents the communication links between the agents. The neighboring set of an agent $i$ is defined as $\mathcal{N}(i) = \{ j \mid (i,j) \in E \}$.

\textcolor{\COLOR}{The adjacency matrix is $\mathbf{A}=[a_{ij}] \in\mathbb{R}^{n\times n}$ with $a_{ij}=1$ for $(i,j) \in E$, otherwise $a_{ij} = 0$. We consider a graph topology that ensures the invertibility of its adjacency matrix. The degree matrix $\mathbf{D}\in\mathbb{R}^{n\times n}$ is given by \textcolor{\COLOR}{$\mathbf{D} = \text{diag}\{d_i\}$, where $d_i$ is the degree of the $i$-th agent.}} 

\subsection{The Lambert W Function}\label{sec3-C}
An effective method widely used for stability analysis is the Lyapunov method, also known as the \textit{Lyapunov-Krasovskii}/\textit{Razumikhin} method for delayed systems. The main advantage of this method is that it applies to most systems, but the stability criteria it provides are conservative. On the other hand, the Lambert W function provides explicit expressions of the characteristic roots in terms of the system parameters but has a limited applicability. The first extensive research of the Lambert W function in the context of stability was \cite{Moore,sunyi}, where the authors reached results related to stability for some specific matrix cases and explored a general theory about the matrix version of the function. \textcolor{\COLOR}{A detailed discussion can be found in  \cite{moradian2022study}}.

\begin{definition} [\hspace{-.1mm}\cite{shinozaki}] 
 \label{Lambert_definition}
    Let $z \in \mathbb{C}$, then the Lambert W function $W_k(z)$ is any function satisfying the equality
    \begin{equation}
    W_k(z)e^{W_k(z)} = z
    \end{equation}
where $k = 0, \pm 1, \pm 2, \hdots, \infty $ signifies the specific branch. 
\end{definition}

The function $W_k(z)$ has an infinite number of branches, but only two of them have real values, $W_0(z)$ and $W_{-1}(z)$. The \textit{principal} branch, $W_0(z)$, is defined on $[-1, +\infty)$ for all $z \in [-1/e, + \infty)$, and $W_{-1}(z)$ is defined on $[-1, -\infty)$ for all $z \in [-1/e, 0)$.

\begin{definition} [\hspace{-.1mm}\cite{Higham}] \label{LambertMatrix_definition}
    \textcolor{\COLOR}{Suppose $\mathbf{H} \in \mathbb{C}^{n \times n}$ is an arbitrary matrix with $\lambda_i(\mathbf{H})=h_i$.} Its Jordan canonical form is $\mathbf{H}=\mathbf{Z} \mathbf{J} \mathbf{Z}^{-1}$. Then the Lambert W matrix function is defined as 
\begin{equation}\label{eq_25}
\mathbf{W}_k(\mathbf{H}) = \mathbf{Z} \mathbf{W}_k(\mathbf{J})\mathbf{Z}^{-1} = \mathbf{Z} \: \textcolor{\COLOR}{\text{diag}} (\mathbf{W}_k(\mathbf{J_i})) \: \mathbf{Z}^{-1},
\end{equation}
\noindent where
{\small{\begin{equation}\label{definition_jordanform}
    \mathbf{W}(\mathbf{J}_i) = 
    \begin{bmatrix}
        W(h_i) & W'(h_i) & \hdots & \frac{W^{(p_i -1)}(h_i)}{(p_i -1)!}\\
             0   & W(h_i) & \ddots &  \vdots\\
       \vdots & & &  W'(h_i) \\
       0 & \hdots& 0 & W (h_i),\\
    \end{bmatrix},
\end{equation}}}
with $p_1 + p_2+ \cdots +p_i = n$, where $p_i$ indicates the order of differentiation. \textcolor{\COLOR}{It is deduced that $\lambda_i(\mathbf{W}_k(\mathbf{H}))=\lambda_i(\mathbf{H})=h_i$.}
\end{definition}

\textcolor{\COLOR}{Definition \ref{Lambert_definition} is equivalently applicable to matrices.} Definition \ref{LambertMatrix_definition} establishes an equality between the eigenvalues of $\mathbf{W}(\mathbf{H})$ and $\mathbf{H}$, but not between $\mathbf{W}(\mathbf{H})$ and the matrix $\mathbf{H}$ itself. The \textcolor{\COLOR}{explicit} correlation \textcolor{\COLOR}{between $\mathbf{W}(\mathbf{H})$ and $\mathbf{H}$} has been researched extensively in specific cases, \textcolor{\COLOR}{none of which include}  multi-agent systems. The authors in  \cite{counter-arg2} studied the function in the scalar case, i.e., as if the system had only one agent. \textcolor{\COLOR}{The authors in \cite{Choudhary_2017} extended the scalar case results to an $n$-th order system; however, their method is limited to systems that can be transformed into companion canonical form.}

\section{Problem Formulation}\label{sec3}

\subsection{System dynamics}
Consider the double integrator multi-agent system, which is given by
\begin{equation}\label{eq_5}
    \begin{cases}
    \Dot{x}_i = v_i,\\
    \Dot{v}_i = u_i,
    \end{cases}
\end{equation}
where $i = 1,2,\hdots,n$ is the agent index. The variables $x_i \in \mathbb{R}$ and $v_i \in \mathbb{R}$ are the position state and the velocity state, respectively. The control input of the $i$-th agent is $u_i$. The consensus protocol is given by
\begin{equation}\label{eq_6}
    {u}_i= - \sum_{j \in \mathcal{N}(i)} a_{ij}(x_i(t) - x_j(t)) + \gamma(v_i(t)-v_j(t))) ,
\end{equation}
where $\gamma \neq 0$ is a parameter representing the weight given to the relative velocity by the control input.

\subsection{Time-Delayed Control Input in Multi-Agent Systems}
Introducing a \textcolor{\COLOR}{constant communication} delay \textcolor{\COLOR}{$\tau$ from agent $i$ to agent $j$}, the consensus protocol is given by
{\small{\begin{equation}\label{eq_6-1}  
    {u}_i= - \sum_{j \in \mathcal{N}(i)} a_{ij}((x_i(t) - x_j(t-\tau)) + \gamma (v_i(t)- v_j(t-\tau))).
\end{equation}}}Let $\mathbf{\Gamma}(t) = [\mathbf{x}(t)^T,\mathbf{v}(t)^T]^T$ where $\mathbf{x}(t) = [x_1, x_2,\hdots, x_n]$ and $\mathbf{v}(t) = [v_1, v_2,\hdots, v_n]$. Then the state-space matrix form of (\ref{eq_5}) can be rewritten using (\ref{eq_6}) as 
\begin{equation}\label{eq_9_10}
\begin{split}
    \mathbf{\Dot{\Gamma}}(t) 
    &= \mathbf{T}  \mathbf{\Gamma}(t) + \mathbf{T}_d \mathbf{\Gamma}(t-\tau) ,
\end{split}
\end{equation}
where $\mathbf{T} = \begin{bmatrix}
        0&\mathbf{I}_n\\
        -\mathbf{D}&-\gamma \mathbf{D}\\
    \end{bmatrix}$, and $\mathbf{T}_d =
     \begin{bmatrix}
        0&0\\
        \mathbf{A}& \gamma \mathbf{A}\\
    \end{bmatrix}  \in \mathbb{R}^{2n \times 2n}.$

\textcolor{\COLOR}{Protocol (\ref{eq_6-1}) is a widely used consensus protocol that has been shown to guarantee consensus convergence, even in the presence of various types of delay, provided that the delay ensures stability,\cite{review,ita}. In other words, the delay bound on $\tau$ that ensures asymptotic stability also ensures convergence to consensus,\cite{Sumpter}.} Before determining the critical delay, it is important to first verify whether the system's stability is dependent on delay from the first place. In \cite{recent}, the authors proved that if $\mathbf{T}$ is stable and the system's characteristic equation satisfies specific conditions, then the system's stability is unaffected by delay. However, whether the stability generally depends on the delay or not remains an NP-hard problem \cite{NP}.

The critical delay of a system is a well-researched problem that has already been solved via multiple methods, but most studies consider only the scalar case (i.e., when there is only one agent, rendering $\mathbf{A}$ and $\mathbf{D}$ scalars), and in the first-order case \cite{kim}, and commensurate delay case \cite{Sumpter}. Applying those methodologies to higher-order systems, such as multi-agents with non-commensurate delay, is a much less studied area of research \cite{recent}.

\section{Stability Analysis}\label{sec4}

\subsection{Lambert W Function Approach}\label{sec4_A}
In \cite{sunyi}, the authors propose a conjecture that when the matrix $\mathbf{T}_d$ does not have repeated zero eigenvalues (has rank $2n-1$), the dominant eigenvalues of the system can be found using the principal branch of $W_k$ at $k=0$. The authors in \cite{counter-arg2} disproved this conjecture with counterexamples. Moreover, they showed that in some cases, it is possible to find \textit{all} characteristic roots of the system using only two branches, $k=0$ and $k=-1$. In either case, not all eigenvalues are needed to analyze stability; it is enough to consider the largest eigenvalue since it acts as an upper bound for the remaining eigenvalues. 

In the system presented in this paper, $\mathbf{T}_d$ has $n$ repeated zero eigenvalues, so the conjecture and its counterargument do not apply, which means a new methodology is needed to apply the Lambert W matrix function approach to this network topology. The stability analysis requires the following lemma:

\begin{lemma}[\hspace{-.1mm}\cite{shinozaki}] \label{lemma_1}
If $z > -1/e$, then
\begin{equation}\label{eq_11}
   \max_{k = 0, \pm 1, \pm 2, \hdots, \pm \infty} \mathrm{R}e[W_k(z)] = \mathrm{R}e[W_0(z)],
\end{equation}
and if $z \leq -1/e$ 
 \begin{equation}\label{eq_12}
    \max_{k = 0, \pm 1, \pm 2, \hdots, \pm \infty} \mathrm{R}e[W_k(z)] = \mathrm{R}e[W_0(z)] =  \mathrm{R}e[W_{-1}(z)].
\end{equation}
\end{lemma}

Suppose that the system (\ref{eq_9_10}) was scalar, i.e., $\mathbf{T} = \alpha$ and $\mathbf{T}_d= \beta$, where $\alpha, \beta \in \mathbb{C}$. Then 
\begin{equation}
    s = \frac{1}{\tau} W(\tau \beta e^{-\tau \alpha}) + \alpha,
\end{equation}
is an explicit expression for the characteristic roots of system (\ref{eq_9_10}) if it was scalar  \cite{shinozaki}. A similar expression was derived for the matrix case, but only when $\mathbf{T}$ and $\mathbf{T}_d$ commute \cite{shinozaki}.

In the case for this paper, when $\mathbf{T}$ and $\mathbf{T}_d$ do not commute, the authors of \cite{sunyi} provided a valid algorithm that can be used. Let $\mathbf{S} =   \mathbf{T} + \mathbf{T}_d e^{-\mathbf{S}\tau}$, then the problem reduces to finding a solution to 
\begin{equation}\label{eq_14}
    \mathbf{S} - \mathbf{T} - \mathbf{T}_d e^{-\mathbf{S}\tau} = 0.
\end{equation}
Since $\mathbf{T}$ and $\mathbf{T}_d$ do not commute, an unknown matrix $\mathbf{Q}_k= 
    \begin{bmatrix}
        \mathbf{Q}_{1}&\mathbf{Q}_{2}\\
        \mathbf{Q}_{3}&\mathbf{Q}_{4}
    \end{bmatrix}$ is introduced that satisfies the following equation
\begin{equation}\label{S_T}
    \tau(\mathbf{S}_k-\mathbf{T})e^{(\mathbf{S}_k-\mathbf{T})\tau} = \tau \mathbf{T}_d \mathbf{Q}_k .
\end{equation}\textcolor{\COLOR}{A new matrix $\mathbf{M}_k$ is introduced such that} $\mathbf{M}_k = \tau \mathbf{T}_d \mathbf{Q}_k$ \textcolor{\COLOR}{(reference \cite{sunyi} explains why this matrix is needed)}. Note that $\mathbf{M}_k$  has the same dimension and form as $\mathbf{T}$ and can be expressed by 
$    \mathbf{M}_k = 
    \begin{bmatrix}
        \mathbf{0}_n&\mathbf{0}_n\\
        \mathbf{M}_{21}&\mathbf{M}_{22}\\
    \end{bmatrix}$.
Using \textcolor{\COLOR}{Definition \ref{Lambert_definition}}, one has
\begin{equation} \label{MK}
    \mathbf{W}(\mathbf{M}_k)e^{\mathbf{W}({\mathbf{M}}_k)} = \mathbf{M}_k,
\end{equation}
so that
\begin{equation}\label{eq_18}
    \mathbf{S_k} = \frac{1}{\tau}\mathbf{W}(\mathbf{M}_k)+\mathbf{T} .
\end{equation}

The eigenvalues of $\mathbf{S}_k$ determine the stability of the system, but $\mathbf{S}_k$ has infinitely many eigenvalues since $k = 0, \pm 1, \pm 2, \hdots, \pm \infty$. \textcolor{\COLOR}{At this stage, Lemma \ref{lemma_1} proves useful in} determining the stability of the system without computing infinite eigenvalues of $\mathbf{S}_k$ for infinite branches of $\mathbf{W}_k$.

\textcolor{\COLOR}{
In order to find  $\mathbf{Q}_k$, (\ref{MK}) is rewritten as
\begin{equation}\label{w_M}
    \mathbf{W}_k(\tau \mathbf{T}_d \mathbf{Q}_k) e^{\mathbf{W}_k(\tau \mathbf{T}_d \mathbf{Q}_k) + \tau \mathbf{T}} =\tau \mathbf{T}_d
\end{equation}
which can be solved using MATLAB. There are infinitely many solutions for $\mathbf{Q}$ satisfying the above equation, which is why an educated initial guess is needed. More is discussed in Section \ref{sec5}.}

\vspace{-3mm}
\subsection{Preliminary Analysis}
In order to find a critical bound on the delay $\tau$, a relationship between the delay and the eigenvalues of $\mathbf{S}_k$ is needed. The branch subscript notation $k$ is removed to make the analysis simpler to follow. The matrix $\mathbf{M}$ depends on $\tau$ by definition, so a relationship between $\mathbf{M}$ and its Lambert W function $\mathbf{W}(\mathbf{M})$ is needed.

\begin{lemma}\label{lemma_2}
The Lambert W function of the matrix $\mathbf{M}$ has the following form
\begin{equation}\label{symmetric_lemma}
     \mathbf{W}(\mathbf{M}) = 
    \begin{bmatrix}
        \mathbf{0}_n & \mathbf{0}_n\\
        \mathbf{W}_{21} &  \mathbf{W}_{22}
    \end{bmatrix},
\end{equation}
where $\mathbf{W}_{21}, \mathbf{W}_{22} \in \mathbb{R}^{n \times n}$ are symmetric and mutually commuting matrices as well as $\mathbf{W}_{21} - \tau \mathbf{D}$ and $\mathbf{W}_{22} - \tau \gamma \mathbf{D}$.

\end{lemma}
\begin{proof}
    Let $\mathbf{W} = \mathbf{W}(\mathbf{M})$ for simplicity, and suppose $         \mathbf{W}=
         \begin{bmatrix}
             \mathbf{W}_{11} &  \mathbf{W}_{12}\\
             \mathbf{W}_{21} &  \mathbf{W}_{22},
        \end{bmatrix},
    $ where each block in $\mathbf{W}$ is an unknown $n\times n$ matrix. Given the structure of $\mathbf{M}_k$, it can be inferred that $\mathbf{W}_{11}=\mathbf{W}_{12}=\mathbf{0}_{n}$ (using $\ref{definition_jordanform}$).
    In (\ref{w_M}), we can write the exponential term as a series expansion: 
    \begin{equation}\label{eq_23}
        \mathbf{W}\sum_{m=0}^{+\infty}\frac{(\mathbf{W}+\mathbf{T})^m}{m!} = \tau \mathbf{T}_d.
    \end{equation}
From the definition of $\mathbf{T}$ and $\mathbf{T}_d$ and by simplifying both sides of (\ref{eq_23}), one can write
{\small{\begin{align}
   \label{matrix_series}
   &\begin{bmatrix}       
       \mathbf{0}_n &  \mathbf{0}_n\\
       \mathbf{W}_{21} &  \mathbf{W}_{22}
   \end{bmatrix}
   \sum_{m=0}^{+\infty} \frac{1}{m!}
   \begin{bmatrix}
       \mathbf{0}_n & \tau\mathbf{I}_n\\
       \mathbf{W}_{21}-\tau\mathbf{D} &  \mathbf{W}_{22}-\tau\gamma \mathbf{D}
   \end{bmatrix}^m \nonumber \\
   &= \tau 
   \begin{bmatrix}
      \mathbf{0}_n&\mathbf{0}_n\\
      \mathbf{A}& \gamma \mathbf{A}
   \end{bmatrix}.
\end{align}}} 
The Cayley-Hamilton theorem allows to write $(\mathbf{W}+\tau\mathbf{T})^m$ \textcolor{\COLOR}{with $m\geq 2n$} as a linear combination of its lower powers:
{\small{ \begin{equation}
\begin{split}\label{eq_27}  (\mathbf{W}+\tau\mathbf{T})^{m}&=\sum_{j=0}^{ \textcolor{\COLOR}{2n-1}} c_{m j} (\mathbf{W}+\tau\mathbf{T})^j\\
    &= \sum_{j=0}^{ \textcolor{\COLOR}{2n-1} } c_{m j} \begin{bmatrix}
             \mathbf{0}_n& \tau\mathbf{I}_n\\
             \mathbf{W}_{21}-\tau\mathbf{D} &  \mathbf{W}_{22}-\tau\gamma \mathbf{D}
         \end{bmatrix}^{\textcolor{\COLOR}{j}},
\end{split}
\end{equation}}}where $c_{m j}$ are constant coefficients.
Using (\ref{eq_27}), the authors of \cite{Moler} prove that the matrix term of the infinite sum in (\ref{matrix_series}) can be expressed as a finite sum 
 \begin{align} \label{eq_28}
e^{(\mathbf{W}+\tau\mathbf{T})}=\sum_{m=0}^{\infty} \frac{(\mathbf{W}+\tau\mathbf{T})^m}{m!} & \equiv \sum_{j=0}^{2n-1} b_j (\mathbf{W}+\tau\mathbf{T})^j,
\end{align}
where $\textcolor{\COLOR}{b_j}$ are strictly positive constants defined by $
\textcolor{\COLOR}{b_j=\frac{1}{j!}+\sum_{m=2n}^{\infty}\frac{c_{mj}}{m!}.}$

The goal of turning (\ref{eq_28}) into a finite sum was to allow the use of proof by induction. It can be proven by induction that (\ref{matrix_series}) holds true if and only if:
\begin{enumerate}
    \item $\mathbf{W}_{21}$ and $\mathbf{W}_{22}$ are symmetric and commute.
    \item The matrices $\mathbf{W}_{21}-\tau\mathbf{D}$ and $\mathbf{W}_{22}-\tau\gamma \mathbf{D}$ commute.
\end{enumerate}
The full proof is omitted here due to space constraints.
\end{proof}

\begin{lemma}\label{Lemma3}
Let $m_i$ represent the $i$-th eigenvalue of $\mathbf{M}_{22}$,
and let $\mathbf{Z} =\begin{bmatrix}
\mathbf{Z}_1&\mathbf{Z}_3\\    \mathbf{Z}_2&\mathbf{Z}_4
    \end{bmatrix} \in \mathbb{R}^{2n\times 2n}$ be an invertible matrix containing the eigenvectors of $\mathbf{W}(\mathbf{M})$.
 Then, $\mathbf{W}(\mathbf{M})$ can be written as 
{\small{\begin{equation}\label{lemma_3}
    \mathbf{W}(\mathbf{M}) = 
    \mathbf{Z}
    \begin{bmatrix}
        \mathbf{0}_n & \mathbf{0}_n\\
        \mathbf{0}_n &  \mathbf{W}(m)
    \end{bmatrix}
    \mathbf{Z}^{-1},
\end{equation}}}where $\mathbf{W}(m)=\mathrm{diag}(W(m_1),W(m_2),\cdots,W(m_n))$.  This $ \mathbf{W}(\mathbf{M})$ has a representation such as (\ref{symmetric_lemma}).

\end{lemma}
\begin{proof}
    Revisiting the branch notation, $\mathbf{M}$ has at least $n$ zero eigenvalues. As $W_k(0)$ is not defined for $k\neq0$, the hybrid branch method from \cite{sunyi} is applied, as in \cite{counter-arg2}. In this approach, whenever there is a zero eigenvalue, the branch is set as $k=0$ so that $W(0)$ is well defined. Next, a matrix is diagonalizable if and only if all of its eigenvalues are distinct; or, if there are repeated eigenvalues, their algebraic and geometric multiplicities must be equal.

    Let $\Lambda_w$ and $\Lambda_m$ represent the \textcolor{\COLOR}{the set of eigenvalues} of $\mathbf{W}(\mathbf{M})$ and $\mathbf{M}$ respectively and let $\lambda(\mathbf{W}_{22})$, $\lambda(\mathbf{M}_{22})$ and $\lambda(\mathbf{0}_n)$ represent the \textcolor{\COLOR}{the set of all eigenvalues} of $\mathbf{W}_{22}$, $\mathbf{M}_{22}$ and $\mathbf{0}_n$. Thus, one can write 
    \begin{equation}
    \begin{split}
        \Lambda_m = \lambda(\mathbf{0}_n) \cup \lambda(\mathbf{M}_{22}), \hspace{3mm}
        \Lambda_w = \lambda(\mathbf{0}_n) \cup \lambda(\mathbf{W}_{22}).
    \end{split}
\end{equation}
This is due to the structure of both matrices in (\ref{MK}) and (\ref{symmetric_lemma}). Given the matrix function definition \cite{Higham}, it is known that the eigenvalues of $\mathbf{W}(\mathbf{M})$ are $W(m_i)$, where $W$ denotes the scalar function from Section \ref{sec3-C}.

Then $\lambda(\mathbf{W}_{22})=\{W(m_1),W(m_2),\cdots,W(m_n)\}$. Since $\mathbf{A}$ is invertible and $\mathbf{M}_{22}$ is proportional to $\mathbf{A}$ ($\mathbf{Q}$ does not affect the invertibility of $\mathbf{A}$), $\mathbf{M}_{22}$ is also invertible and therefore does not have any zero eigenvalue. The algebraic multiplicity $\mu(0)$ of zero eigenvalues of $\mathbf{W}(\mathbf{M})$ is then $\mu(0)=n$, so $\mu(m) = 2n-n=n$ is the algebraic multiplicity of the nonzero eigenvalues. Then, the geometric multiplicity of $0$ is
\begin{equation}
\begin{split}
    \gamma(0) &= \mathrm{Nullity}(\mathbf{W}) \\ & = 2n - \mu(m) = \mu(0).
\end{split}
\end{equation}

The matrix $\mathbf{W}_{22}$ is symmetric, and if it has repeated eigenvalues their algebraic and geometric multiplicities are equal, i.e., $\gamma(W(m_i))=\mu(W(m_i))$ if $m_i$ is not distinct. We can conclude that $\mathbf{W}(\mathbf{M})$ is diagonalizable and can be written as (\ref{lemma_3}).  This completes the proof.
\end{proof}

\begin{lemma}
The matrix $\mathbf{W}(\mathbf{M})$ has the following form:
{\small{\begin{equation}
     \mathbf{W}(\mathbf{M})= \begin{bmatrix}
        \mathbf{0}_n & \mathbf{0}_n\\
        - \mathbf{W}(\mathbf{M}_{22})\mathbf{Z}_1\mathbf{Z}_2^{-1}&  \mathbf{W}(\mathbf{M}_{22})
    \end{bmatrix}.
    \end{equation}}}
\end{lemma}

\begin{proof}
From (\ref{symmetric_lemma}), and by multiplying (\ref{lemma_3}) by $\mathbf{Z}$ on the left, one can write
\begin{equation}\label{eq_33}
    \begin{bmatrix}
        \mathbf{0}_n & \mathbf{0}_n\\
        \mathbf{W}_{21} &  \mathbf{W}_{22}
    \end{bmatrix}
     \begin{bmatrix}
       \mathbf{Z}_1&\mathbf{Z}_3\\
       \mathbf{Z}_2&\mathbf{Z}_4
    \end{bmatrix}
    =
    \begin{bmatrix}
       \mathbf{Z}_1&\mathbf{Z}_3\\
       \mathbf{Z}_2&\mathbf{Z}_4
    \end{bmatrix}
     \begin{bmatrix}
        \mathbf{0}_n & \mathbf{0}_n\\
        \mathbf{0}_n &  \mathbf{W}(m)
    \end{bmatrix}.
    \end{equation}
   Simplifying (\ref{eq_33}):
 {\small{      \begin{equation}
     \begin{bmatrix}     \mathbf{0}_n&\mathbf{0}_n\\\mathbf{W}_{21}\mathbf{Z}_1+\mathbf{W}_{22}\mathbf{Z}_2&\mathbf{W}_{21}\mathbf{Z}_3+\mathbf{W}_{22}\mathbf{Z}_4
        \end{bmatrix}
        =
        \begin{bmatrix}
            \mathbf{0}_n&\mathbf{Z}_3 \mathbf{W}(m)\\
            \mathbf{0}_n&\mathbf{Z}_4 \mathbf{W}(m)
        \end{bmatrix}.
    \end{equation}}}Comparing the left hand-side and right hand-side of the equality, the following is valid
   {\small{ \begin{equation}\label{eq31}
        \begin{split}
           \mathbf{W}_{21}\mathbf{Z}_1 &= -\mathbf{W}_{22}\mathbf{Z}_2,\\
        \mathbf{Z}_3 \mathbf{W}(m) & = \mathbf{0}_n,\\
        \mathbf{W}_{22}\mathbf{Z}_4&=\mathbf{Z}_4 \mathbf{W}(m) .
        \end{split}
    \end{equation}}}
    Note that \textcolor{\COLOR}{the second Equation in (\ref{eq31})} is only valid if $\mathbf{Z}_3=\mathbf{0}_n$ since $\mathbf{W}(m)$ is a diagonal matrix. Additionally, since $\mathbf{Z}$ is invertible, then
\begin{equation*}
    \mathrm{det}(\mathbf{Z}) = \mathrm{det}(\mathbf{Z_1}\mathbf{Z_4})=\mathrm{det}(\mathbf{Z_1})\mathrm{det}(\mathbf{Z_4}) \neq 0.
\end{equation*}
Therefore $\mathbf{Z_1}$ and $\mathbf{Z_4}$ must both be invertible, \textcolor{\COLOR}{so the last equation in (\ref{eq31}) can be rewritten as}
\begin{equation}
\begin{split}
        \mathbf{W}_{22} &= \mathbf{Z_4} \mathbf{W}(m) \mathbf{Z_4}^{-1}=\mathbf{W}(\mathbf{M}_{22}),
    \end{split}
\end{equation}   
and $\mathbf{W}_{21}= -\mathbf{W}(\mathbf{M}_{22})\mathbf{Z}_2\mathbf{Z}_1 ^{-1}$.
This completes the proof.
\end{proof}

\subsection{Critical Delay Bound}
To summarize the previous lemmas, the matrix $\mathbf{W}(\mathbf{M})$ from (\ref{w_M}) can be written as
\begin{equation}
    \mathbf{W}(\mathbf{M})=
     \begin{bmatrix}
        \mathbf{0}_n & \mathbf{0}_n\\
        - \bar{{\mathbf{{W}}}}(\mathbf{M}_{22})&  \mathbf{W}(\mathbf{M}_{22})
    \end{bmatrix},
\end{equation}
where $\bar{{\mathbf{{W}}}}(\mathbf{M}_{22})=\mathbf{W}(\mathbf{M}_{22})\mathbf{Z}_1\mathbf{Z}_2^{-1}$ and $\mathbf{W}(\mathbf{M}_{22})$ are symmetric. The purpose of the notation $\bar{{\mathbf{{W}}}}(\mathbf{M}_{22})$ is to remember that $\mathbf{W}_{21}$ is a function of $\mathbf{W}(\mathbf{M}_{22})$ and consequently of $\mathbf{W}_{22}$. This result is similar to the scalar case studied in \cite{counter-arg2}. It is also intuitive since the two lower blocks of the system matrix $\mathbf{T}_d$ are proportional.  Consequently, (\ref{eq_18}) can be written as 
\begin{equation}
    \mathbf{S}=
  \begin{bmatrix}
             \mathbf{0}_n & \mathbf{I}_n\\
            -\frac{1}{\tau}\bar{{\mathbf{{W}}}}(\mathbf{M}_{22})- \mathbf{D} &  \frac{1}{\tau}\mathbf{W}(\mathbf{M}_{22})-\gamma \mathbf{D}
         \end{bmatrix}.
\end{equation}
Next, a correlation between the eigenvalues of $\mathbf{S}$ and the eigenvalues of $\mathbf{W}(\mathbf{M})$ and $\mathbf{D}$ is needed. 

For the rest of the paper, the set of eigenvalues of all matrices of interest are sorted in descending order, i.e., $\lambda_1\geq\lambda_2\geq\hdots\geq\lambda_i$.

\begin{lemma} \label{lemma_5}
    Let $\lambda_i(\frac{1}{\tau}\mathbf{W}(\mathbf{M}_{22})-\gamma\mathbf{D})=\textcolor{\COLOR}{\eta_i}$ and $\lambda_i(-\frac{1}{\tau}\bar{{\mathbf{{W}}}}(\mathbf{M}_{22})-\mathbf{D})=\textcolor{\COLOR}{\mu_i}$. Then, \textcolor{\COLOR}{the multi-agent system (\ref{eq_5}) with the control input (\ref{eq_6-1})} is asymptotically stable if and only if
    \begin{equation} \label{condition}
       \textcolor{\COLOR}{\eta_1}<0,  \quad  \text{and}  \quad \textcolor{\COLOR}{\mu_1}<0.
    \end{equation}
 
\end{lemma}

\begin{proof}
To find the eigenvalues, start by writing the characteristic equation. Using \cite[Theorem 2]{Silvester_2000},
{\footnotesize{
\begin{equation}
\begin{aligned}\label{det}
    \mathrm{det}(s\mathbf{I}_{2n}-\mathbf{S}_k)= \mathrm{\textcolor{\COLOR}{det}} \Big{(}
    \begin{bmatrix}
        s\mathbf{I}_n & -\mathbf{I}_n\\
             \frac{1}{\tau}\bar{{\mathbf{{W}}}}(\mathbf{M}_{22})+\mathbf{D} &  s\mathbf{I}_n-\frac{1}{\tau}\mathbf{W}(\mathbf{M}_{22})+\gamma \mathbf{D}
    \end{bmatrix}\Big{)}\\
    =\mathrm{det}(s^2\mathbf{I}_n-s(\frac{1}{\tau}\mathbf{W}(\mathbf{M}_{22})-\gamma\mathbf{D})-(-\frac{1}{\tau}\bar{{\mathbf{{W}}}}(\mathbf{M}_{22})-\mathbf{D})) =0. &
\end{aligned}
\end{equation}}}Authors in (\hspace{1sp}\cite[Theorem 2.1]{Tao}) state that the eigenvalues of the sum of two Hermitian matrices that commute are the sum of the eigenvalues of their sum. Then, from Lemma \ref{lemma_2},
\begin{equation}
\begin{split}
\lambda_i((\frac{1}{\tau}\mathbf{W}(\mathbf{M}_{22})-\gamma \mathbf{D})-(-\frac{1}{\tau}\bar{{\mathbf{{W}}}}(\mathbf{M}_{22})-\mathbf{D}))
 &= \textcolor{\COLOR}{\eta_i} - \textcolor{\COLOR}{\mu_j},
\end{split}
\end{equation}
where $\textcolor{\COLOR}{\mu_j}$ is some permutation of the eigenvalues $\textcolor{\COLOR}{\mu_i}$, for $j\in\{1,\hdots,n\}$.

As demonstrated in \cite{Ren}, (\ref{det})  can be written as
\begin{equation}\label{eq_36}
   p(\mathbf{S})= \prod_{i=1}^n (s^2-\textcolor{\COLOR}{\eta_i}s-\textcolor{\COLOR}{\mu_i})=0,
\end{equation}
where $p(\mathbf{S})$ is the characteristic polynomial of $\mathbf{S}$. So, the eigenvalues of $\mathbf{S}$ are
    \begin{equation}
        s_i = \frac{\textcolor{\COLOR}{\eta_i} \pm \sqrt{\textcolor{\COLOR}{\eta_i}^2+4\textcolor{\COLOR}{\mu_i}}}{2}.
    \end{equation}
Using the Routh-Hurwitz criterion, $\mathrm{Re}(s_i)<0$ if and only if condition (\ref{condition}) is satisfied. This concludes the proof.
\end{proof}

\ifLongVersion
\textcolor{\COLOR}{
\begin{remark}\label{Remark_1}
The stability conditions established in Lemma \ref{lemma_5} are not merely technicalities but constitute the very mechanism by which consensus is achieved in our multi-agent framework. Specifically, by enforcing
\begin{equation}
\begin{split}
\max_i \lambda_i\!\Bigl(-\tfrac{1}{\tau}\,\overline{W}(M_{22}) - D\Bigr) < 0 \quad\text{and} \\ 
\max_i \lambda_i\!\Bigl(\tfrac{1}{\tau}W(M_{22}) - \gamma D\Bigr) < 0,
\end{split}
\end{equation}we ensure the asymptotic stability of the consensus dynamics. Consequently, all agents’ state differences decay exponentially to zero, and the network converges to a common trajectory—i.e., consensus is reached. Hence, the stability proof of Lemma \ref{lemma_5} directly underpins and guarantees the consensus objective of the manuscript.
\end{remark}}
\vspace{.2cm}
\else
\fi

\textcolor{\COLOR}{In order to relate the delay to stability, it is necessary to relate the condition (\ref{condition}) to the individual eigenvalues of $ \mathbf{\bar{W}}(\mathbf{M}_{22})$, $\mathbf{W}(\mathbf{M}_{22})$, and $\mathbf{D}$. The next theorem establishes a direct correlation between the numerical value of the delay and eigenvalues $\eta_i$ and $\mu_i$ on which stability depends (Lemma \ref{lemma_5}).}

\begin{theorem} \label{theorem}
    Let $\lambda_i( \mathbf{W}(\mathbf{M}_{22})) = w_i$ and $\lambda_i( \mathbf{\bar{W}}(\mathbf{M}_{22})) = \tilde{w_i}$. Then, a necessary condition for $\mathrm{Re}[s_i] < 0$ is
    \begin{equation}\label{necessary}
        \begin{cases}
            \frac{1}{\tau} w_i  < \gamma d_j,\\
            -\frac{1}{\tau}\tilde{w}_i< d_j,
        \end{cases}
    \end{equation}
    for any permutation $(i,j)\in\{1,\hdots,n\}$, such that $i+j=1+n$.
\end{theorem}

\ifLongVersion
\begin{proof}
If condition (\ref{condition}) is satisfied, it is necessary that the lower bounds of $\textcolor{\COLOR}{\eta_i} ,\textcolor{\COLOR}{\mu_i}$ also satisfy (\ref{condition}).
The rest of the proof relies heavily on a combination of Weyl's inequality and Horn's conjecture. The conjecture was recently proved by Terence Tao and others \cite{honey}, making it a theorem.  
The lower bounds on $\textcolor{\COLOR}{\eta_1}$ and $\textcolor{\COLOR}{\mu_1}$ are 
\begin{equation}
    \begin{split}
       \frac{1}{\tau} w_i-\gamma d_j \leq \textcolor{\COLOR}{\eta_1},\\
         -\frac{1}{\tau}\tilde{w}_i - d_j \leq \textcolor{\COLOR}{\mu_1},
    \end{split}
\end{equation}
where $i+j =1+n$, with $i,j\in\{1,\hdots,n\}$ \cite{taotopics}. Then, the following condition is imposed on the individual eigenvalues
\begin{equation}
    \begin{cases}
        \frac{1}{\tau} w_i<\gamma d_j,\\
        -\frac{1}{\tau}\tilde{w}_i< d_j.
    \end{cases}
\end{equation}
This concludes the proof.
\end{proof}
\else
\begin{proof}
The complete proof is given \cite{Arxiv}.
\end{proof}
\fi

Another interesting result is the upper bound that can be imposed on the eigenvalues $s_i$ using Weyl's standard inequality (\hspace{1sp}\cite[Equation (1.54)]{taotopics}):
\textcolor{\COLOR}{An equivalent converse of Theorem \ref{theorem} can also be proven, i.e., if $\tau$ satisfies the corresponding condition, then there is stability.} 

Back to the branch notation, if $m_i$ are real, then Lemma \ref{lemma_1} can be used to find the maximum eigenvalues $w_1$, using only $k=0,-1$. Since $\bar{{\mathbf{{W}}}}(\mathbf{M}_{22})$ and $\mathbf{{W}}(\mathbf{M}_{22})$ are symmetric, then their eigenvalues are all real. The only real branches of the Lambert W function are $k=0,-1$, so these eigenvalues must correspond to either $W_0$ or $W_{-1}$. This implies that $\lambda(\mathbf{M}) \in \mathbb{R}$, so Lemma \ref{lemma_1} can be used to find the maximum eigenvalues. Additionally, $\lambda(\mathbf{M})=m_i$ is in the domain of the $W_k$, for $k=0,-1$.

The critical delay $\tau^*$ therefore corresponds to either $W_0(m_1)$ or $W_{-1}(m_1)$, depending on the numerical value of $m_1$. Note that $W_0(m_1)$ and $W_{-1}(m_1)$ are functions of $\tau$ since $m_i$ is directly proportional to $\tau$. Finally, it is worthwhile to make the following remarks:
\begin{enumerate}
    \item If $m_i\in[-1/e, 0)$, then as $\tau\to 0$, $w_i=W_{-1}(m_i)\to -\infty$, ensuring the first condition of (\ref{necessary}) is not violated.
    \item If $m_i\in [-1/e, +\infty)$, then as $\tau\to +\infty$, $W_{0}(m_i)\to +\infty$, violating the first condition of (\ref{necessary}), rendering the system not stable. And as $\tau \to 0$, $W_{0}(m_i)\to 0$ ensuring the second condition of (\ref{necessary}) is not violated.
\end{enumerate} 
Clearly, there is a direct correlation between the stability of the system and the time delay. However, Lemma \ref{lemma_5} is needed to find the critical delay $\tau^*$.

\section{Numerical Example}\label{sec5}

This section presents a numerical example to validate the theoretical results. The parameter $\gamma$ is set to $1$. \textcolor{\COLOR}{Different topologies with different numbers of agents are considered to showcase the scalability of the approach.}

\begin{figure}[htbp]
    \centering
    \begin{subfigure}[b]{0.21\textwidth}
        \includegraphics[width=1\textwidth]{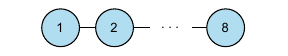}
        \caption{}
        \label{fig:cycle}
    \end{subfigure}
    \
    \begin{subfigure}[b]{0.21\textwidth}
        \includegraphics[width=1\textwidth]{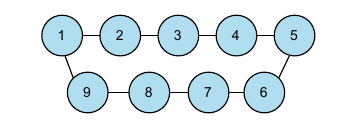}
        \caption{}
        \label{fig:line}
    \end{subfigure}
    \vspace{1em}  
    \begin{subfigure}[b]{0.21\textwidth} 
        \centering
        \includegraphics[width=1\textwidth]{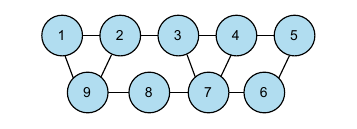}
        \caption{}
        \label{fig:rand}
    \end{subfigure}\vspace{-2mm}
    \caption{Three examples of connected graphs. \vspace{-2mm}}
    \label{fig:two_rows}
\end{figure}Fig. \ref{fig:two_rows} shows three different network topologies exhibiting an invertible adjacency matrix.

The first step is to find the matrix $\mathbf{Q}$ from (\ref{w_M}).  The authors in \cite{counter-arg2} provided an effective method to find a correct solution using \texttt{fsolve} in MATLAB. The need for a \textit{posteriori} analysis is common for this type of delay (\hspace{1sp}\cite{RAM20091940}). The method is based on finding a structured initial guess for \texttt{fsolve} such that the algorithm converges to the correct solution. The initial guess $\mathbf{Q}_{0k}$ is found by reverse engineering the steps listed in Section \ref{sec4_A}. \ifLongVersion{} \else \textcolor{\COLOR}{A pseudocode summarizing the steps in finding the delay bound $\tau^*$ is included in the extended version \cite{Arxiv}}.\fi

\ifLongVersion
\begin{algorithm}[h]
\caption{Iterative Eigenvalue-Based Stability Verification}
\label{alg:stability}
\KwIn{Adjacency matrix $\mathbf{A}$, degree matrix $\mathbf{D}$ scalar parameter $\gamma$}
\KwOut{Matrix $\mathbf{Q}_k$, matrix $\mathbf{W}_k$ satisfying desired eigenvalue conditions, final $\tau^*$}

\textbf{Step 1: Initialization} \\
\quad Compute $\mathbf{T}$ and $\mathbf{T}_d$ from $\mathbf{A}$, $\mathbf{D}$ and $\gamma$, and set $\Delta \tau$ = 0.001.

\textbf{Step 2: Eigenvalue Estimation} \\
\quad Use the QpmR algorithm~\cite{Zitek} to compute eigenvalues of the system defined in (\ref{eq_9_10}), using an arbitrary initial guess for $\tau$. 

\textbf{Step 3: Initial Matrix Construction} \\
\quad Construct $\mathbf{S}_0$ from eigenvalues obtained in Step 2. \\
\quad Construct $\mathbf{Q}_0$ using $\mathbf{S}_0$.

\textbf{Step 4: Root Finding} \\
\quad Use MATLAB's \texttt{fsolve} to solve for $\mathbf{Q}_k$ using $\mathbf{Q}_0$ as the initial guess.

\textbf{Step 5: Iterative Stability Search} \\
\While{$\eta_1 < 0$ \textbf{and} $\mu_1 < 0$}{
    Compute $\mathbf{M}_k = \tau \mathbf{T}_d \mathbf{Q}_k$. \\
    Construct matrix $\mathbf{Z}$ from eigenvectors and eigenvalues of $\mathbf{M}_k$. \\
Compute {\footnotesize{$\mathbf{W}(m) = \mathrm{diag}(W(m_1), W(m_2), \cdots, W(m_n))$}}. \\
Compute $\mathbf{W}(\mathbf{M}) = 
\mathbf{Z}
\begin{bmatrix}
    \mathbf{0}_n & \mathbf{0}_n\\
    \mathbf{0}_n &  \mathbf{W}(m)
\end{bmatrix}
\mathbf{Z}^{-1}$. \\

Compute $\bar{\mathbf{W}}(\mathbf{M}_{22}) = \mathbf{W}(\mathbf{M}_{22}) \mathbf{Z}_1 \mathbf{Z}_2^{-1}$. \\

Evaluate $\mu_1 = \max_i \lambda_i\left(-\frac{1}{\tau} \bar{\mathbf{W}}(\mathbf{M}_{22}) - \mathbf{D}\right)$. \\

Evaluate $\eta_1 = \max_i \lambda_i\left(\frac{1}{\tau} \mathbf{W}(\mathbf{M}_{22}) - \gamma \mathbf{D}\right)$. \\
    Update $\tau = \tau + \Delta \tau$.
}
\end{algorithm}

The above algorithm is applied to all three different topologies in Fig. \ref{fig:two_rows}. 
\else
\fi

\textcolor{\COLOR}{
To verify the validity of the proposed approach, a comparison is done against the delay bound found in \cite{Hou}. The system considered in \cite{Hou} is the same as (\ref{eq_5}), and the consensus protocol is the same as (\ref{eq_6}), except that in \cite{Hou}, the authors consider a commensurate delay (same delay on all agents including agent $i$). 
We denote our delay bound by $\bar{\tau}_1$, \ifLongVersion found using Algorithm \ref{alg:stability} and \else \fi we use $\bar{\tau}_2$ to denote the delay bound found in \cite{Hou}.}

\vspace{-4mm}
\textcolor{\COLOR}{\begin{table}[h] 
\centering 
\caption{\textcolor{\COLOR}{Comparison of the allowable time-delay bounds.
 }}\label{Tau_table}
 \begin{tabular}{|c|c c c|} 
 \hline
Delay bound  & Fig. \ref{fig:cycle} & Fig. \ref{fig:line} & Fig. \ref{fig:rand}\\
\hline
$\bar{\tau}_1 (\text{Lemma} \ \ref{lemma_5}) $ & 0.362   &  0.339 &  0.326 \\ 
$\bar{\tau}_2$ (\cite{Hou}) & 0.331 & 0.333 & 0.270\\
 \hline
 \end{tabular}
 \end{table}}

\textcolor{\COLOR}{It is observed from Table \ref{Tau_table} that the delay bounds are close in value, validating the results further.
Next, as a way to measure the energy efficiency of the control law, the following energy function is considered
\begin{equation} \label{energy_func}
    \rho (t)=\sum_{i=1}^{n} u^2_i(t).
\end{equation}
The function (\ref{energy_func}) can be used as a performance index (\hspace{1sp}\cite{kiarash}) to evaluate the efficiency for each of control law (\ref{eq_6-1}) and the control law from \cite{Hou}.}

\textcolor{\COLOR}{Let us define the energy index $\nu$ as follows:
 \begin{equation}
     \nu = \int_{0}^{t_f} \rho (t) dt.
 \end{equation}
To quantify performance, $\nu$ is computed for the three topologies in Fig. ~\ref{fig:two_rows} with $t_f =6$. Let $\nu_1$ and $\nu_2$ correspond to the proposed control law (\ref{eq_6-1}) and the method in \cite{Hou}, respectively. Table~\ref{nu_table} presents the results for each method.}
\begin{table}[h]
\centering
\caption{\textcolor{\COLOR}{ Energy index $\nu$ for $\tau = 0.26$, $\gamma = 1$, comparing our method and \cite{Hou} across three topologies.}}
\label{nu_table}
\textcolor{\COLOR}{
\begin{tabular}{|c|c c c|} 
 \hline
 $\nu$ & Fig. \ref{fig:cycle} & Fig. \ref{fig:line} & Fig. \ref{fig:rand} \\
 \hline
 $\nu_1$ (\ref{eq_6-1}) & 12.95 & 12.62 & 13.85 \\ 
 $\nu_2$ \cite{Hou}& 50.48 & 49.58 & 168.63 \\
 \hline
\end{tabular}
}
\end{table}

\ifLongVersion
Fig. \ref{fig:energy_plot} compares the energy functions of our proposed method with that of \cite{Hou}. The results presented in Table \ref{nu_table} and Fig. \ref{fig:energy_plot} demonstrate that our control law achieves significantly lower energy consumption. It is worth noting that we made minimal adjustments to the control gain parameters, selecting $k_1 = 1$ and $k_2 = 1$ in \cite{Hou}'s method.
\begin{figure}[htbp]
\centerline{\includegraphics[scale=.5]{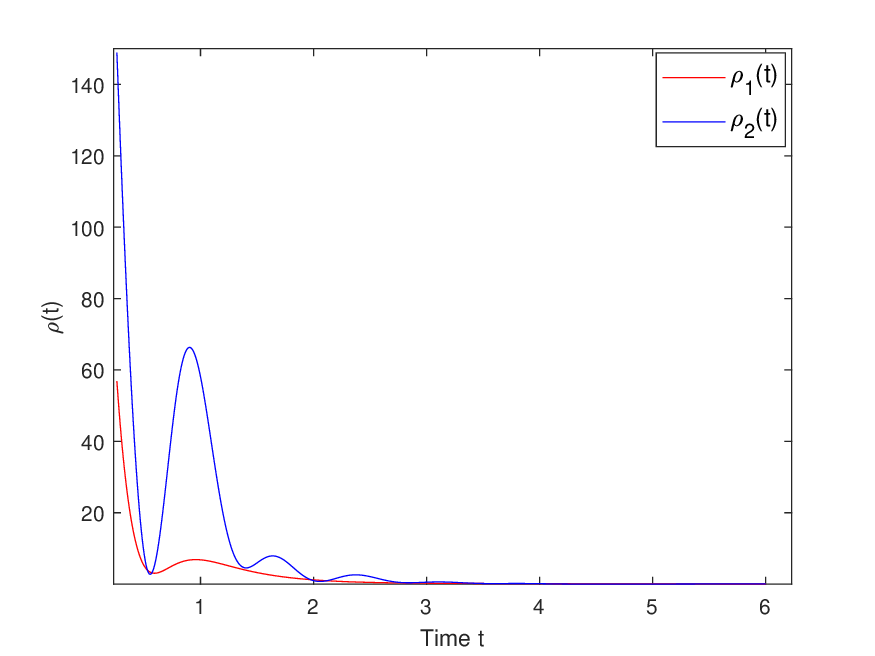}}
\caption{The energy function (\ref{energy_func}) for $\tau=0.26$, $\gamma=1$, corresponding to topology \ref{fig:cycle} ($\rho_1$: control input (\ref{eq_6-1}), $\rho_2$: \cite{Hou}'s method).}
\label{fig:energy_plot}
\end{figure}
\else
\textcolor{\COLOR}{The extended version \cite{Arxiv} includes an energy function comparison with \cite{Hou}.}
\fi

\textcolor{\COLOR}{Although the derived stability delay bound is similar to that found in earlier work based on commensurate delays \cite{Hou,yang}, a key distinction lies in the energy consumption of the control law. With the assumption of a non-commensurate delay, the proposed approach demonstrates a significant reduction in overall energy consumption during convergence. Quantitatively, the total control effort—measured as the integral of the squared control input over time—is markedly lower across all agents. This suggests that accounting for non-commensurate delays is not only more practical, but also improves resource efficiency, making the approach more suitable for energy-constrained multi-agent systems such as autonomous vehicles or sensor networks.}

\section{Conclusion}\label{sec6}
In this paper, an analytical framework was developed to address the impact of constant time delay in weakly connected multi-agent systems, specifically modeled by double integrator dynamics. By leveraging the Lambert W matrix function, an upper bound on the time delay that ensures system stability was derived, even in the presence of non-commensurate communication delays. The rigorous theoretical analysis provides a solid foundation for the results. Numerical simulations have further validated the effectiveness of the proposed approach in real-world scenarios involving leaderless networks of agents. This work contributes to the broader understanding of stability in multi-agent systems with time delays and lays the groundwork for future research on more complex delay models and larger-scale systems. Specifically, it would be of interest to generalize the applicability of the proposed approach to a broader class of multi-agent system topologies, including those characterized by directed graphs and singular adjacency matrices. Moreover, the approach has the potential to analyze stability subject to parameters other than delay.


\bibliography{sn-bibliography}
\bibliographystyle{ieeetr}

\end{document}